\documentclass[journal]{IEEEtran}

\usepackage{multicol}
\usepackage{graphicx}
\usepackage[latin1]{inputenc}
\usepackage{amsmath}
\usepackage{amsfonts}
\usepackage{acronym}
\usepackage{amssymb}
\usepackage{dsfont}    
\usepackage{epsfig}
\usepackage{ifthen}
\usepackage{psfrag}
\usepackage{url}
\usepackage{listings}
\usepackage{comment}
\usepackage{fancyhdr}
\usepackage[newitem,newenum]{paralist}
\usepackage{booktabs}
\usepackage{tabularx}
\usepackage{color}
\usepackage{setspace}
\usepackage[normalsize]{subfigure}

\definecolor{lightgray}{rgb}{0.9,0.9,0.9}
\definecolor{black}{rgb}{0,0,0}

\graphicspath{{/home/awillig/lehre/vl/stochastische-lb/ws2005_2006/figures/}}

\ifthenelse{\isundefined{\chapter}}{}{

\newtheorem{theorem}{Theorem}[chapter]

\newtheorem{myexamplethm}{Example}[chapter]

\newtheorem{myexcursionthm}{Excursion}[chapter]

\newcommand{\printsolution}{true}
\newtheorem{myproblemthm}{Problem}[chapter]

\newcommand{\problem}[4]{
  \begin{myproblemthm}[#1]\label{#2}
    #3
   \end{myproblemthm}
   \begin{footnotesize}
     \ifthenelse{\equal{\printsolution}{true}}{{\bf Solution:}#4}{}
   \end{footnotesize}
}

\newtheorem{myinlineproblemthm}{Inline problem}[chapter]

\newcommand{\inlineproblem}[4]{

  \noindent
  \begin{description}
  \item \begin{myinlineproblemthm}[#1]\label{#2}
      \begin{center}
        \fbox{
          \begin{minipage}{0.9\textwidth}
            \begin{footnotesize}
              {\normalfont #3}

              \vspace{1cm}

              \ifthenelse{\equal{\printsolution}{true}}{{\bf Solution:}#4}{}

            \end{footnotesize}
          \end{minipage}}
        \end{center}
   \end{myinlineproblemthm}
  \end{description}
}
}

\newcounter{sheetprobnumber}
\stepcounter{sheetprobnumber}

\input epsf

\makeatletter
\def\Corr{\mathop{\operator@font Corr}\nolimits}
\def\Cov{\mathop{\operator@font Cov}\nolimits}
\def\Expect{\mathop{\operator@font E}\nolimits}
\def\Var{\mathop{\operator@font Var}\nolimits}
\makeatother

\newcommand{\fixme}[2]{\ifthenelse{\equal{\printsolution}{true}}{\textbf{[#1]}\footnote{#2}}{}}

\newcommand{\ve}[1]{{\mathbf{#1}}}

\newcommand{\abs}[1]{\ensuremath{\left|#1\right|}}

\newcommand{\E}[1]{\mathds E\{{#1}\}}

\lstdefinelanguage[gnuplot]{C}
{keywords={set,xlabel,ylabel,logscale,output,plot,using,title,terminal,boxes,w},
  sensitive=false,
  morecomment=[l]{\#},
  morecomment=[s]{/*}{*/},
  morestring=[b]*,
}

\newtheorem{proposition}{\textbf{Proposition}}
\newtheorem{problem}{Problem}

\begin{document}
\title{Cooperative Secret Communication with Artificial Noise in Symmetric Interference Channel}

\author{{Jingge Zhu, Jianhua Mo and Meixia Tao, \IEEEmembership{Member,~IEEE}}
\thanks{Manuscript received 28-Jun-2010. The associate editor
coordinating the review of this letter and approving it for
publication was J. Jalden.} \thanks{The authors are with Dept. of
Electronic Engineering, Shanghai Jiao Tong University, P. R. China.
Emails: \{{zhujingge}, {mjh}, {mxtao}\}@sjtu.edu.cn.}
\thanks{This work was supported in part by the NSF of China under
grant 60902019 and Shanghai Pujiang Talent Program under grant
09PJ1406000.} }

\vspace{-0.5cm}

\maketitle

\begin{abstract}
We consider the symmetric Gaussian interference channel where two
users try to enhance their secrecy rates in a cooperative manner.
Artificial noise is introduced along with useful information. We
derive the power control and artificial noise parameter for two
kinds of optimal points, max-min point and single user point. It is
shown that there exists a critical value $P_c$ of the power
constraint, below which the max-min point is an optimal point on the
secrecy rate region, and above which time-sharing between
single user points achieves larger secrecy rate pairs. It is also
shown that artificial noise can help to enlarge the secrecy rate
region, in particular on the single user point.
\end{abstract}

\maketitle


\section{Introduction}

The problem of secret communication is considered 
in the seminal paper of Wyner \cite{Wyner:75}. It is shown that
perfect secrecy can be achieved without any key, provided that the
receiver has a better channel than the eavesdropper. Recently, the
secret communication in wireless networks has been intensively
studied for various scenarios. Broadcast channel with confidential
message is considered in \cite{Liang:08} whereas multiple-access
channels with secrecy constraint is studied in \cite{Tekin:08} and
\cite{Liang_mac:08}.  The secrecy rate region of Gaussian
interference channel with an external eavesdropper is investigated
in \cite{Koyluoglu:09}.

In this work, we consider the secret communication in a two-user
symmetric interference channel as shown in Fig.~\ref{fig:system
model}
where each receiver has to decode its own message while eavesdropping on the other's message.
It is first pointed out in \cite{GoelNegi08} that by introducing
artificial noise in the transmission along with the useful
information, the secrecy rate region can be enlarged as the
artificial noise causes additional interference to the eavesdropper.
The key idea in our work is that although the two users in this
system do not
 trust each other because both can potentially eavesdrop on the other's message,
 nevertheless, they can enhance their secrecy rates in a cooperative manner.
 It is called \emph{semi-secret} in \cite{Yates:08} as the achieved secret
 communication depends on trusting other transmitters.

We derive the optimal power control and artificial noise parameter
for two different points on the secrecy rate region, namely, max-min
point and single user point. We show that depending on the
relationship of power constraint and channel conditions, both points
can potentially achieve optimal secrecy rate pairs. A criterion is
given explicitly. We also show that while it is not helpful in
improving the max-min point, artificial noise can enlarge the
secrecy rate on single user point.

     \begin{figure}[t!]
       \centering
       \includegraphics[scale=0.42]{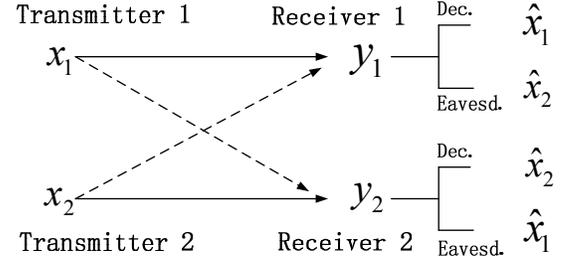}
       \caption{Interference channel with confidential messages.}
       \label{fig:system model}
       \vspace{-0.2cm}
   \end{figure}

Notation $(\cdot)^T$ denotes transpose. $\E{x}$ stands for
expectation of random variable $x$. The function $\log(\cdot)$ is
taken to the base 2.

\section{Symmetric Interference Channel with Artificial Noise}

The symmetric interference channel can be modeled as
\begin{subequations}
\begin{align}
 &y_1=\sqrt{a}x_1+\sqrt{a_c}x_2+n_1\\
 &y_2=\sqrt{a_c}x_1+\sqrt{a}x_2+n_2
\end{align}
\label{eq:scalarchannel}
\end{subequations}
where $a$ and $a_c$ are the gains for the direct channels and cross
channels, respectively, the transmitted signal $x_i$, for $i=1, 2$,
is subject to the peak power constraint $\E{\abs{x_i}^2}=p_i\leq P$,
and $n_i$ is the additive white Gaussian noise with zero mean and
variance $N$. The transmitted signal $x_i$ is composed of message
$s_i$ and artificial noise $z_i$, i.e.\ $x_i=s_i+z_i$. Here $z_i$ is
chosen to be Gaussian hence cannot be decoded by any receiver. We
split the transmission power as $\E{s_i^2}=(1-\lambda_i)p_i$ and
$\E{z_i^2}=\lambda_ip_i$, with parameter $\lambda_i\in[0,1]$.

The formal definition of secrecy rate for the interference channel
can be found in \cite{LiuMaric:08}, which considers the case where
only one user sends artificial noise. We generalize their work to
allow both users to use artificial noise and define the achievable
secrecy rate region under our setting as follows.

Let $(R_1,R_2)$ denote a rate pair satisfying
(\ref{cons:region_scalar}).
\begin{figure*}
\begin{subequations}
\begin{align}
&0\leq R_1\leq R_1^s:=\log\left(1+\frac{a(1-\lambda_1)p_1}{N+a_cp_2+a\lambda_1p_1}\right)-\log\left(1+\frac{a_c(1-\lambda_1)p_1}{N+a_c\lambda_1p_1+a\lambda_2p_2}\right) \label{cons:region_1}\\
&0\leq R_2\leq
R_2^s:=\log\left(1+\frac{a(1-\lambda_2)p_2}{N+a_cp_1+a\lambda_2p_2}\right)-\log\left(1+\frac{a_c(1-\lambda_2)p_2}{N+a_c\lambda_2p_2+a\lambda_1p_1}\right)
\label{cons:region_2}
\end{align}
\label{cons:region_scalar}
\end{subequations}
\hrulefill
\end{figure*}
Then the secrecy rate region is the union of all possible rate pairs
$(R_1,R_2)$ for transmission power $0\leq p_i\leq P$ and power
splitting parameter $0\leq\lambda_i\leq1$, $i=1,2$.

The achievability of the secrecy rate region defined above can be
justified using the similar stochastic encoding argument in
\cite{LiuMaric:08}.  If time-sharing between transmission strategies
is allowed, the convex hull of the above rate region can also be
achieved.   We can interpret the constraint
(\ref{cons:region_scalar}) as follows: In any working system, the
receiver can always decode its own message successfully by
considering interferences as pure noise, then it subtracts the
already decoded message and try to decode the message from the other
transmitter. The difference on the right-hand side of each of the
two constraints is the maximum amount of information one can hide
from the other, from an information theoretic point of view.

Note that nonnegative secrecy rate only exists when the direct
channel is stronger than the cross channel. This can be verified
directly with the expression of $R_1^s$ or $R_2^s$. So we only
consider the case where $a>a_c$ hereafter.

\section{Main Results}
In this section, we present the main results on the optimal power
allocation $\{p_i,\lambda_i\}_{i=1}^2$ on two points of the secrecy
rate region, namely, max-min point and single user point.

\subsection{Max-min Point}

We first define an optimal point in the following sense:
\begin{align*}
R_{min}^*:=\max_{\{\lambda_i,p_i\}} \min\{R_1, R_2\}.
\end{align*}
Note that since the second inequality in (\ref{cons:region_1}) and
(\ref{cons:region_2}) can be tight simultaneously, the above
definition is equivalent to $R_{min}^*:=\max_{\{\lambda_i,p_i\}}
\min\{R_1^s, R_2^s\}$.

\begin{proposition}
For interference channel (\ref{eq:scalarchannel}),
$R_{min}^*=\log\left(\frac{(a+a_c)^2}{4aa_c}\right)$ with
$\lambda_1^*=\lambda_2^*=\lambda^*$, where $\lambda^*$ can be chosen
arbitrarily from the interval
$[0,\frac{a_c}{a}-\frac{N(a-a_c)}{Pa(a+a_c)}]$, and
$p_1^*=p_2^*=p^*=\frac{N(a-a_c)}{(a+a_c)(a_c-a\lambda^*)}$. Among
the maximizing points, $\lambda^*=0$ yields the minimum transmission
power $p^*_{min}=\frac{N(a-a_c)}{a_c(a+a_c)}$.  If the power
constraint $P<p^*_{min}$, the maximizing points are  $\lambda^*=0$
and $p^*=P$. \label{prop:maxmin}
\end{proposition}

\begin{proof} It is intuitive to see that in order to achieve the
point $R_{min}^*$, we need $p_1=p_2=p$ and
$\lambda_1=\lambda_2=\lambda$ because of the competitive nature of
the two users. This will be justified at the end of the proof. We
will now maximize
\begin{equation*}
R_1^s=R_2^s=R^s=\log\frac{(N+a_cp+ap)(N+a_c\lambda p+a\lambda
p)}{(N+a_cp+a\lambda p)(N+a\lambda p+a_cp)}.
\end{equation*}
The maximum value of $R^s$ should satisfy:
\begin{equation}
\frac{\partial{R^s}}{\partial p}=0,\
\frac{\partial{R^s}}{\partial\lambda}=0 \label{eq:critical_pt}
\end{equation}
It is found that for arbitrary $\lambda$, choosing
\begin{equation}
p(\lambda)=\frac{N(a-a_c)}{(a+a_c)(a_c-a\lambda)} \label{eq:opt_p}
\end{equation}
always forms a solution to (\ref{eq:critical_pt}). It can also be
shown that the second-order derivatives of $R^s$ are negative at
these points, i.e.,\ they are all maximizing points of the function.
Taking the constraints $\lambda\in[0,1]$ and $p(\lambda)\in[0,P]$
into consideration, we see that the valid value of $\lambda^*$
should be in the interval
$[0,\frac{a_c}{a}-\frac{N(a-a_c)}{Pa(a+a_c)}]$, and the optimal
$p^*$ is obtained by substituting $\lambda^*$ into (\ref{eq:opt_p}).
Note that the possible maximizing points on the boundary
($\lambda=0$ for example, which cannot be found by solving the
equations (\ref{eq:critical_pt})) are also included in the solution.

The minimum required transmission power maintaining the secrecy rate
$R_{min}^*$ is $p^*_{min}=\frac{N(a-a_c)}{a_c(a+a_c)}$ by choosing
$\lambda^*=0$. In the case where $P<p^*_{min}$, there is no solution
to (\ref{eq:critical_pt}) satisfying the constraint
$\lambda\in[0,1]$, and the $R_{min}^*$ is achieved by $\lambda^*=0$
and $p^*=P$ since $R^s$ is now increasing with $p$ and decreasing
with $\lambda$.

We now justify that the same transmission power $p$ and power splitting parameter $\lambda$ are indeed required to achieve $R_{min}^*$. Define $\ve A=\nabla R^s_1(\lambda^*,p^*)\nabla R_2^{sT}(\lambda^*,p^*)$, where $\nabla R^s_i=[\partial R^s_i/\partial p_1,\partial R^s_i/\partial p_2,\partial R^s_i/\partial \lambda_1,\partial R^s_i/\partial \lambda_2]^T$ is the gradient of $R^s_i$, and $\nabla R^s_i(\lambda^*,p^*)$ means $\nabla R^s_i$ evaluated at the point $\lambda_1=\lambda_2=\lambda^*$, $p_1=p_2=p^*$. It is clear that $\ve A$ has only one eigenvalue which is equal to its trace, given by
 \begin{equation*}
tr(\ve A)=\frac{\partial R^s_1}{\partial p_1}\frac{\partial
R^s_2}{\partial p_1}+\frac{\partial R^s_1}{\partial
p_2}\frac{\partial R^s_2}{\partial p_2}+\frac{\partial
R^s_1}{\partial \lambda_1}\frac{\partial R^s_2}{\partial
\lambda_1}+\frac{\partial R^s_1}{\partial \lambda_2}\frac{\partial
R^s_2}{\partial \lambda_2}.
\end{equation*}

Straightforward calculation shows that $\frac{\partial
R^s_i}{\partial p_j}(\lambda^*, p^*)$
 is negative for $i\neq j$ and positive for $i=j$. Also, $\frac{\partial R^s_i}{\partial \lambda_j}(\lambda^*, p^*)$
 is positive for $i\neq j$ and negative for $i=j$. So $tr(\ve A)$ is always negative. Therefore, $\ve A$ is negative definite hence $\ve d^T\ve A\ve d<0$ for any $\ve d\neq \ve 0$. Note that $\ve d^T\ve A\ve d$ can also be rewritten as
\begin{equation}
 \nabla R_1^{sT}(\lambda^*,p^*)\ve d\cdot\nabla R_2^{sT}(\lambda^*,p^*)\ve
 d<0.
 \label{ineq}
\end{equation}
Inequality (\ref{ineq}) means that any deviation from the optimal
points will decrease the value of either $R^s_1$ or $R^s_2$,
thereby, the minimum of the two becomes smaller. In other words, the
deviated point cannot be a max-min point. Thus, our choices of $p$
and $\lambda$ are validated and the proposition is proved.
\end{proof}

\subsection{Single User Point}

We now investigate another point, called \textit{single user point},
on which one user tries to maximize its own secrecy rate with the
help of the other user, i.e. $R_{su,i}^*=\max R_i^s$. It is clear that due to the symmetry, we have $R_{su,1}^*=R_{su,2}^*=R_{su}^*$. We will show that
through this kind of cooperation, one user obtains an appreciably
large secrecy rate while the secrecy rate of the other is zero. In
addition, we also find that through time-sharing, we can achieve
larger rate pairs than the max-min point when $P$ is greater than a
critical value.

\begin{proposition}
The single user point $R_{su}^*$ is obtained with $(\lambda_1^*=0$,
$\lambda_2^*=1, p_1^*=P$, $p_2^*=\frac{\Delta-N}{a+a_c})$ or
$(\lambda_1^*=1$, $\lambda_2^*=0$, $p_1^*=\frac{\Delta-N}{a+a_c}$,
$p_2^*=P)$ and it is given by
\begin{align}
&R_{su}^*=\max R_1^s=\max R_2^s  \nonumber\\
&=\log\left(\frac{(aN+a_{c}\Delta)(a_{c}N+a\Delta)+a(a+a_{c})(a_{c}N+a\Delta)P}{(aN+a_{c}\Delta)(a_{c}N+a\Delta)+a_{c}(a+a_{c})(aN+a_{c}\Delta)P}\right)
\label{eq:single user}
\end{align}
with $\Delta=\sqrt{N^2+(a+a_c)NP}$. \label{prop:single_user}
\end{proposition}

\begin{proof} Without loss of generality, we analyze the single user point for user 1 only. From (2a), $R_1^s$ is decreasing with
$\lambda_1$  and increasing with both $\lambda_2$ and $p_1$. Hence,
to maximize $R_1^s$ we should have $\lambda_1^*=0$, $\lambda_2^*=1$
and $p_1=P$. Substituting them into $R_1^s$ and solving the equation
$\frac{\partial R_1^s}{\partial p_2}=0$ for $p_2$, we find
$p_2^*=\frac{\Delta-N}{a+a_c}<\frac{P}{2}$. It can again be
justified with the second-order derivative that it is indeed the
maximized point. \end{proof}

We now compare the secrecy rate pairs achieved by max-min point and
single user point. At max-min point, each user gets the same secrecy
rate $R_{min}^*$, given in Proposition \ref{prop:maxmin}. By equal
time-sharing between the two single user points, each user gets the
same secrecy rate $R_{su}^*/2$, where $R_{su}^*$ is given in
Proposition \ref{prop:single_user}.

\begin{proposition}
When the power constraint $P$ is larger than the critical power $P_c
=
\frac{N(a-a_{c})(a^{2}+a_{c}^{2}+6aa_{c})}{(a_{c}^{2}+3aa_{c})^{2}}$,
equal time-sharing between single user points achieves larger rates
than the max-min point, otherwise, the max-min point achieves larger
rates.
\end{proposition}

\begin{proof} This proposition can be easily proved by solving $R_{su}^* = 2R_{min}^*$ and using the monotonicity of
$R_{su}^*$. \end{proof}

\section{Numerical Examples and Discussions}

   \begin{figure}[tbp]
    \centering
        \includegraphics[scale=0.5]{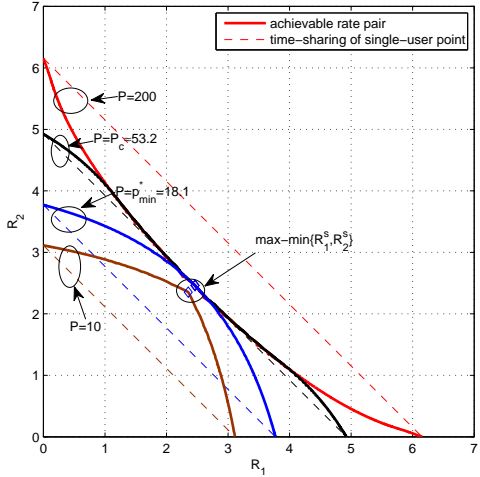}
    \caption{Achievable secrecy rate region for different power constraint
                at $a = 1, a_c =0.05$ and $N = 1$. The critical power $P_c\approx53.2$ and $p_{min}^{*}\approx18.1$.}
    \label{fig:diff power constraint}
    \vspace{-0.2cm}
  \end{figure}
Fig. \ref{fig:diff power constraint} demonstrates some numerical
results on the achievable secrecy rate region for different power
constraints with fixed channel condition. When $P=P_c$, the max-min
point (diamond in the figure) is the same point obtained by equal
time-sharing of two single user points, otherwise there are
significant gaps between the achievable rates of two methods. When
$P\geq p_{min}^*$, the max-min points of different $P$ coincide.
These results verify our analytical findings in Proposition 1.

   \begin{figure}[tbp]
    \centering
        \includegraphics[scale=0.5]{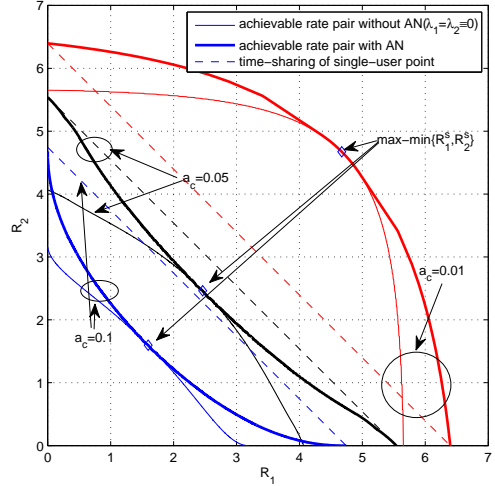}
    \caption{Achievable secrecy rate region for different channel conditions
                at $a=1$, $N=1$, $P=100$.}
    \label{fig:diff channel condition}
\vspace{-0.2cm}

   \end{figure}
Fig. \ref{fig:diff channel condition} shows the benefit of the
artificial noise (AN).
The secrecy rates with and without artificial noise ($\lambda_1=\lambda_2\equiv0$) are plotted for fixed $P$ but
different channels. We observe that the rate region with artificial noise is always larger than that without artificial noise, in particular, artificial noise increases the secrecy rate achieved on single user point.
For $a_c=0.01$, the point $R_{min}^*$ is the optimal point and is superior than points achieved by time-sharing.
For larger $a_c$, the optimal points are obtained by time-sharing.

The above results show that the secrecy rate region in cooperative
symmetric interference channel with artificial noise behaves
significantly different from the classical capacity in symmetric
interference channel (\cite{EtikTseWang08}). When classical capacity
is concerned, the max-min point is always attained when the sum rate
$R_1+R_2$ is also maximized. However, for secrecy capacity, the
point $\max(R_1^s+R_2^s)$ does not necessarily coincide with
$R_{min}^*$ all the time.

%
%

\bibliographystyle{IEEEtran}
\bibliography{Tao_CL2010-1121}

\end{document}